\documentclass[11pt]{amsart}

\usepackage{amssymb}
\usepackage{amsmath,amsthm,amsfonts}  
 \usepackage{graphicx}
\usepackage{amssymb}
\usepackage{pdfsync}
\usepackage{fancyhdr}
\usepackage[applemac]{inputenc}
\usepackage{mathabx}
\usepackage{mathtools}
\usepackage{subfig}
\usepackage{bbm}

\usepackage[T1]{fontenc}

\usepackage{color}
\usepackage{xcolor}
\usepackage[breaklinks,colorlinks,backref]{hyperref}
\hypersetup{
    colorlinks, 
    linktoc=all, 
    linkcolor=red,  
  citecolor=hyptxt,
  urlcolor=blue}
  \hypersetup{
  citebordercolor=,
  filebordercolor=green,
  linkbordercolor=blue
}
\definecolor{hyptxt}{rgb}{0.7, 0.4, 0.9}

\newcommand\symbolwithin[2]{%
  {\mathmakebox[\widthof{\ensuremath{{}#2{}}}][c]{{#1}}}}

\newtheorem{thm}{Theorem}
\newtheorem{assum}[thm]{Assumption}
\newcommand{\belem}{\begin{lemma}}
\newcommand{\enlem}{\end{lemma}}
\newtheorem{defn}{Definition}
\newtheorem{prop}[defn]{{Proposition}}
\newtheorem{lemma}[defn]{{Lemma}}

\newtheorem{comm}{COMMENT}
\newtheorem{rem}[defn]{Remark}
\newtheorem{ex}[defn]{Example}
\newcommand{\betheo}{\begin{thm}}
\newcommand{\entheo}{\end{thm}}

\definecolor{hervecolor}{rgb}{0.8,0,0.7}

\newcommand{\be}{\begin{equation}}
\newcommand{\en}{\end{equation}}
\newcommand{\bea}{\begin{eqnarray}}
\newcommand{\ena}{\end{eqnarray}}
\newcommand{\beano}{\begin{eqnarray*}}
\newcommand{\enano}{\end{eqnarray*}}
\newcommand{\bee}{\begin{enumerate}}
\newcommand{\ene}{\end{enumerate}}

\newcommand{\beprop}{\begin{prop}}
\newcommand{\enprop}{\end{prop}}
\newcommand{\berem}{\begin{rem}$\!\!${\bf }$\;$\rm }
\newcommand{\beex}{\begin{ex}$\!\!${\bf }$\;$\rm }
\newcommand{\enex}{ \end{ex}}
\newcommand{\enrem}{ \end{rem}}
\newcommand{\becomm}{\begin{comm}$\!\!${\bf }$\;$\rm }
\newcommand{\encomm}{ \end{comm}}

\newcommand{\ket}[1]{|\kern.3ex#1\kern.3ex\rangle}
\newcommand{\bra}[1]{\langle\kern.3ex #1 \kern.3ex|}
\newcommand{\scalar}[2]{\langle\kern.3ex #1 \kern.3ex|\kern.3ex#2\kern.3ex\rangle}

\newcommand{\mc}{\mathcal}
\newcommand{\ii}{\mathsf{i}}

\newcommand*\interior[1]{\mathring{#1}}
\newcommand{\ip}[2]{\langle{#1},{#2}\rangle}
\newcommand{\mb}{\mathbb}
\newcommand{\sob}{W^{1,2}({\mb R})}

\newcommand{\ltwo}{L^2({\mb R})}

\newcommand{\ds}{\displaystyle}

\def\calH{{\mathcal H }}

\def\D{\mathcal{D}}

\def\R{\mathbb{R}}

\def\C{\mathbb{C}}

\def\lg{\langle }
\def\rg{\rangle }

\def\ud{\mathrm{d}}

\def\sfP{\mathsf{P}}

\def\mFs{\mathfrak{F_s}}

\def\vz{\pmb{0}}
\def\bI{\mathbbm{1}}

\newcommand{\se}{_{\scriptscriptstyle E}}
\newcommand{\ses}{_{\scriptscriptstyle E,\sigma}}
\newcommand{\sez}{_{\scriptscriptstyle E,0}}

\numberwithin{equation}{section}

\begin{document}
\date{\today}
 
\title[Weyl-Heisenberg quantization for interval]{Regularized quantum motion in a bounded set: Hilbertian aspects }
\author[Bagarello-Gazeau-Trapani]{
 Fabio Bagarello$^{\mathrm{a},\mathrm{b}}$, Jean-Pierre Gazeau$^{\mathrm{c}}$, and Camillo Trapani$^{\mathrm{d}}$}
 
 \address{\emph{$^{\mathrm{A}}$ Dipartimento di Ingegneria,}\\
\emph{Universit\`a di Palermo, I-90128  Palermo, Italy}} 

 \address{\emph{$^{B}$ INFN, Sezione di Catania, Italy}}

\address{\emph{$^{\mathrm{C}}$ Universit\'e Paris Cit\'e}\\
\emph{CNRS, Astroparticule et Cosmologie, F-75013 Paris}} 

\address{\emph{$^{\mathrm{D}}$ Dipartimento di Matematica e Informatica,}\\
\emph{Universit\`a di Palermo, 90123  Palermo, Italy}}

\email{e-mail:
fabio.bagarello@unipa.it, gazeau@apc.in2p3.fr, camillo.trapani@unipa.it}

{\abstract{It is known that the momentum operator canonically conjugated to the position operator  for a particle moving in some bounded interval of the line {(with Dirichlet boundary conditions)  is not essentially self-adjoint}: it has a continuous set of self-adjoint extensions. We prove that essential self-adjointness can be recovered by symmetrically weighting the momentum operator with a  positive bounded  function approximating the indicator function of the considered interval.  This weighted momentum operator  is consistently obtained from a similarly weighted classical momentum through the so-called Weyl-Heisenberg covariant integral quantization of functions or distributions.   
 }}

\maketitle

\subsubsection*{Keywords: selfadjointness, Weyl-Heisenberg group, integral quantization}

\subsubsection*{MSC: 81Q10, 81S30, 81R05, 81R30}

\tableofcontents

\section{Introduction}
\label{intro}

The aim of this paper is to study functional analysis aspects of quantum and semi-classical descriptions of the motion of a particle in a bounded  Borel set in $\R$, e.g.,  an interval, when these quantum models are derived from the Weyl-Heisenberg covariant integral quantization \cite{bergaz14,gazeau18,becugaro17,gazkoi19,gazhumoze20}. Briefly, the latter procedure is the linear map
\begin{equation}
\label{qmap0}
f(q,p) \mapsto A_f = \int_{\R^2}  \frac{\ud q \,\ud p}{2\pi }\,f(q,p)\, e^{\ii (pQ-qP)}\mathfrak{Q}_0 e^{-\ii (pQ-qP)}\, , 
\end{equation}
transforming the function (or distribution) $f(q,p)$ of phase space variables $(q,p)\in \R^2$ into the operator $A_f$ acting in the Hilbert space $L^2(\R,\ud x)$. The operator $\mathfrak{Q}_0$ has unit trace, and the operator $e^{\ii (pQ-qP)}\equiv U(q,p)$ is the familiar unitary Weyl operator built form the self-adjoint position $Q$ and momentum $P$ operators.   
The map \eqref{qmap0} transforms $f=1$ into the identity operator $\bI$, and  is Weyl covariant:
\begin{equation*}
\label{covtrans10}
U(q_0,p_0)\,A_f \,U(q_0,p_0)^{\dag}= A_{\mathcal{T}(q_0,p_0)f}\, , \quad \left(\mathcal{T}(q_0,p_0)f\right)(q,p):= f\left(q-q_0,p-p_0\right)\,. 
\end{equation*} 
This quantization allows a lot of freedom: arbitrariness in the choice of $\mathfrak{Q}_0$, and wide range of objects $f$ to be quantized (e.g. from smooth functions to tempered distributions).
In particular, if one deals with the  motion of a particle, mass $m$,  in an interval $E$, nothing prevents us to restrict the map \eqref{qmap0} to functions mainly localised on $E$, more precisely through the weighting 
{\begin{equation}
\label{weightu}
f(q,p)\mapsto u_E(q)f(q,p)\, , 
\end{equation}
 where $u_E$ is the indicator function of $E$ or a positive smooth approximation of it. Hence the kinetic energy of the particle, $K= \dfrac{p^2}{2m}$ is transformed into a variable mass one,  $u_E(q)K= \dfrac{p^2}{2m(q)}$, $m(q)= \dfrac{m}{u_E(q)}$, which becomes infinite outside $E$. { In this sense, the mechanical model is different from the infinite square well for which the potential $V$,  and consequently the mechanical energy $K+V$, are infinite outside the well.}} 

{Also note that our regularisation procedure is distinct from well-known approaches  in use in quantum mechanics in dealing with singular potentials, for instance those connected with  point interactions \cite{seba86,albev88,brasche94}.}
  
Now, the  noticeable outcome of the present work is the proof of the following (Theorem \ref{prop17}).

\textit{Let $a(x)$ be a positive bounded function on $\R$ and $P= -\ii \dfrac{\ud}{\ud x}$ the (essentially) self-adjoint momentum operator in $L^2(\R,\ud x)$. Then the symmetric $a(x)Pa(x)$ is self-adjoint. }
  
This looks quite a reasonable result, but its proof is not given in the literature, in our knowledge, at least in the way we prove it.
{ Actually, products of functions of the position operator $Q$ and of the momentum $P$ have been considered in the literature not so often and not so recently and mostly as tools for facing other problems like in \cite{arthurs}. For instance, B. Simon in \cite{simon} considers the problem of when an operator of the form $f(Q)g(P)$, with $f,g$ bounded functions, is trace class. No general results for the self-adjointness of $f(Q)g(P)f(Q)$ when $f,g$ are possibly unbounded, but reasonable functions, are know to us.}

From a physical side, the interest of this result {(Theorem 5) }is the following: through the quantization \eqref{qmap0}, the weight function $u_E$ gives rise to the function $a(x)$. {The latter is meant to replace (and regularize) the boundaries of  $E$, to get something smoother.} Surprisingly enough, this regularizing procedure creates other several mathematical subtleties, which are those on which we focus (and solve) in this paper. For this reason we believe that this paper could be interesting both for its mathematical and for its physical aspects.

The material of our quantization procedure and the resulting semi-classical portraits are summarized in Section \ref{CIQWH}.   The application to the quantization of weighted or  truncated observables for the motion in  bounded or semi-bounded sets  and subsequent semi-classical portraits is implemented in Section \ref{qutrunc}. Hilbertian aspects and analysis of relevant operators are developed in Section \ref{windof}.  Specific examples are examined in Section \ref{modmom}. Our results and subsequent questions   are indicated  in Section \ref{conclu}.

\section{Weyl Heisenberg covariant integral quantization of the motion on the  line: a survey}
\label{CIQWH}

In this section, we give an outline of  the Weyl Heisenberg covariant integral quantization, and the subsequent  semi-classical portrait, of the motion on the line for which the phase space of the pairs (position $q$, momentum $p$ within the context of Hamiltonian mechanics, time $b$, frequency  $\omega$ within the context of Signal analysis) is the Euclidean plane $\R^2= \{q,p\}$ (resp.  $\R^2= \{b,\omega\}$) equipped with the   Lebesgue measure $\ud q\,\ud p$ (resp. $\ud b\,\ud \omega$). All justifications and details are found in \cite{bergaz14,gazeau18,becugaro17}. In this paper we keep the notations $(q,p)$ with no considerations of physical dimensions, hence the Planck constant $\hbar$ is put equal to $1$. 

Note that  the mathematical content of this section is not really not new. For instance one can refer to the  contributions by Werner \cite{werner84}, Feichtinger and Kozek \cite{feichwer98},  Luef and Skrettingland \cite{luefskrett18}. 
We should also refer to  the pioneer works by Klauder \cite{klauder63,klauder64}, Berezin \cite{berezin75}, Daubechies and Grossmann \cite{daubgross80,daubgrossreig83} on the integral quantization based on standard coherent state. 

\subsection{Quantization}
 Precisely, we transform a function $f(q,p)$ into an operator $A_f$ in some separable Hilbert space $\mathcal{H}$ through a linear map which sends the function $f=1$ to the identity operator in $\mathcal{H}$ and which respects the basic translational symmetry of the phase space $ \R^2$.   From now on $\mathcal{H}= L^2(\R,\ud x)\equiv L^2(\R)$. Let $Q$ and $P$ be the essentially self-adjoint position and momentum operators defined on the Schwartz space (their common core  domain) $\mathcal{S}(\R) \subset \mathcal{H}$ as $Q\phi(x) =x \phi(x)$ and  $P\phi(x) =-\ii \,\dfrac{\ud }{\ud x}\phi(x)$ respectively, with $[Q,P] = \bI$. The Weyl or displacement operator is the unitary operator defined by
\begin{equation}
\label{Uqp}
U(q,p) = e^{\ii (pQ-qP)}\, , \quad U^{-1}(q,p)= U^{\dag}(q,p)= U(-q,-p)\,,
\end{equation}
with the alternative ``disentanglement'' formulae,  
\begin{equation}
\label{ }
U(q,p)= e^{-\ii\frac{qp}{2}}\,e^{\ii pQ}\,e^{-\ii qP}= e^{\ii\frac{qp}{2}}\,e^{-\ii qP}\,e^{\ii pQ}\, , 
\end{equation}
and the composition property,
\begin{equation}
\label{UUeU}
U(q,p)\,U(q^{\prime},p^{\prime})= e^{-\frac{\ii}{2}(qp^{\prime}-pq^{\prime})}\,U (q+q^{\prime},p+p^{\prime})
\end{equation}
From the above one sees that the map $(q,p) \mapsto U(q,p)$ is a non commutative unitary projective representation of the abelian translation group $\R^2$. It is just the non trivial  part of the von Neumann unitary irreducible representation of the Weyl-Heisenberg (WH)  group 
$G_{\mathcal{WH}}$
\begin{align}
\label{WHUIR}
G_{\mathcal{WH}}&\ni (\zeta,q,p) \mapsto e^{\ii\zeta}U(q,p)\, , \\ 
\label{WHCR} (\zeta,q,p)(\zeta^{\prime},q^{\prime},p^{\prime})&= \left(\zeta + \zeta^{\prime} - \frac{1}{2}(qp^{\prime}-pq^{\prime}),q + q^{\prime},p+p^{\prime}\right)\,.
\end{align}
Given a unit trace operator $\mathfrak{Q}_0$ on $\mathcal{H}$, $\mathrm{Tr}(\mathfrak{Q}_0)=1$, its unitarily displaced versions resolve the identity in  $\mathcal{H}$ (in a weak sense):
 \begin{equation}
\label{resunit}
 \int_{\R^2} \frac{\ud q \,\ud p}{2\pi }\, \mathfrak{Q}(q,p)= \bI\,, \quad \mathfrak{Q}(q,p):= U(q,p)\mathfrak{Q}_0U^{\dag}(q,p)\,.
\end{equation}
For instance, for the one-rank projector $\mathfrak{Q}_0=|\psi\rg\lg\psi|$, $\Vert\psi\Vert=1$, one obtains
\begin{equation}
\label{resunitcs}
 \int_{\R^2} \frac{\ud q \,\ud p}{2\pi}\, |q,p\rg_{\psi}{}_{\psi}\lg q,p|= \bI\,, \quad |q,p\rg_{\psi}:= U(q,p)|\psi\rg\,,
\end{equation}
i.e. the familiar resolution of the identity  by the overcomplete family of \textit{coherent states} $|q,p\rg_{\psi}$. 

The fundamental property \eqref{resunit} results from the application of the Schur's Lemma. It allows to define the  integral quantization of a function (or tempered distribution, see for details \cite{bergaz14}) on the phase space  as  the linear map
\begin{equation}
\label{qmap}
f(q,p) \mapsto A_f = \int_{\R^2}  \frac{\ud q \,\ud p}{2\pi }\, f(q,p)\, \mathfrak{Q}(q,p)
\end{equation}
whenever this integral exists in a weak sense. 
Hence  the identity  $\bI$  is the quantized version of the function $f=1$ while $\mathfrak{Q}_0$ is the quantum version of the Dirac peak $2\pi\delta(q,p)$.

The map is covariant in the following sense:
\begin{equation}
\label{covtrans1}
U(q_0,p_0)\,A_f \,U(q_0,p_0)^{\dag}= A_{\mathcal{T}(q_0,p_0)f}\, , 
\end{equation} 
where 
\begin{equation}
\label{covtrans2}
\quad \left(\mathcal{T}(q_0,p_0)f\right)(q,p):= f\left(q-q_0,p-p_0\right)\,.
\end{equation}
This property justifies the  name ``Weyl-Heisenberg covariant quantization'': no point in the phase space is privileged. 
 
 Let us now introduce the \textit{WH transform} of the operator $\mathfrak{Q}_0$
\begin{equation}
\label{WHtr}
\Pi(q,p) = \mathrm{Tr}\left(U(-q,-p)\mathfrak{Q}_0 \right)\, . 
\end{equation}
This defines a bounded function on the phase space, $\Vert \Pi \Vert_{\infty}=1$, such as $\Pi(0,0)=1$. It can be viewed as an apodization \cite{apodi}  on the plane, which determines the extent of our coarse graining  of the phase space.  In a certain sense this  function  corresponds to the Cohen ``$f$'' function \cite{cohen66} (for more details see \cite{cohen13} and references therein) or to Agarwal-Wolf filter functions \cite{agawo70}. 

The \textit{inverse WH-transform} to \eqref{WHtr} exists due to the two remarkable properties \cite{bergaz14,becugaro17} of the displacement operator $U(q,p)$,
\begin{equation}
\label{IWHtr}
\int_{\R^2} U(q,p) \,\frac{\ud q\, \,\ud p}{2\pi } = 2{\sf P}\ \mbox{and}\ \mathrm{Tr}\left(U(q,p)\right)= 2\pi\delta(q,p)\,,
\end{equation}
where ${\sf P}= {\sf P}^{-1}$ is the parity operator defined as ${\sf P}U(q,p){\sf P}= U(-q,-p)$. Consistently to \eqref{IWHtr}, its trace is put equal to $\mathrm{Tr}({\sf P})=1/2$.
Hence one gets  $\mathfrak{Q}_0$ from $\Pi$:
\begin{equation}
\label{PiQ0}
 \mathfrak{Q}_0 = \int_{\R^2} \frac{\ud q\,\ud p}{2\pi}\,  \Pi(q,p)\, U(q,p) \, . 
\end{equation}
 
Many features of our quantization procedure are better captured if one uses an  alternative quantization formula through the so-called symplectic Fourier transform. The latter is defined for $f\in L^1(\R)$ or, more generally,  for $f$  a tempered distribution, i.e. in $\mathcal{S}^{\prime}(\R^2)$,  as 
  \begin{equation}
\label{symFourqp}
 \mathfrak{F_s}[f](q,p)= \int_{\R^2}e^{-\ii (qp^{\prime}-q^{\prime}p)}\, f(q^{\prime},p^{\prime})\,\frac{\ud q^{\prime}\,\ud p^{\prime}}{2\pi \hbar} \, . 
\end{equation}
 It is involutive, $\mathfrak{F_s}\left[\mathfrak{F_s}[f]\right]=  f$ like its \textit{dual} defined 
 as $\overline{\mathfrak{F_s}}[f](q,p)= \mathfrak{F_s}[f](-q,-p)$. 
  
By replacing $\mathfrak{Q}_0$ in \eqref{qmap} by its expression \eqref{PiQ0} in terms of $\Pi(q,p)$ one easily find the alternative  form of  the quantization map:
\begin{equation}
\label{quantPi1}
f\mapsto A_f= \int_{\R^2} \frac{\ud q\,\ud p}{2\pi} \,  \overline{\mFs}[f](q,p)\, \Pi(q,p) U(q,p) \, .  
\end{equation}
Note that if $\mathfrak{Q}_0$ is symmetric, i.e. $\overline{\Pi(-q,-p)}= \Pi(q,p)$,  a real function $f(q,p)$ is mapped to a symmetric operator $A_f$. Moreover, if  the unit trace-class $\mathfrak{Q}_0$ is  non-negative, i.e., is a density operator, then  a real semi-bounded  function $f(q,p)$ is mapped to a self-adjoint operator $A_f$ through the Friedrich extension \cite{akhglaz81} of its associated semi-bounded quadratic form.

The  formula \eqref{quantPi1} allows to make more easily explicit the action $A_f$ as the  integral operator  
\begin{equation}
\label{Afint}
L^2(\R,\ud x) \ni \phi(x) \mapsto (A_f \phi)(x)  = \int_{-\infty}^{+\infty}\mathrm{d} x^{\prime}\, \mathcal{A}_f(x,x^{\prime})\, \phi(x^{\prime})\, , 
\end{equation}
with  kernel given by
\begin{equation}
\label{AfintK}
 \mathcal{A}_f(x,x^{\prime})= \frac{1}{2\pi}\int_{-\infty}^{+\infty}\mathrm{d} q\, \widehat{f}_{p}(q,x^{\prime}-x)\, \widehat{\Pi}_{p}\left(x-x^{\prime},q- \frac{x+x^{\prime}}{2}\right)\,.
\end{equation}
Here the symbol $\widehat{f}_{p}$ stands for partial Fourier transform of $f$ with respect its second variable $p$:
\begin{equation}
\label{pFTp}
\widehat{f}_{p}(q,y)= \frac{1}{\sqrt{2\pi}}\int_{-\infty}^{+\infty}\ud p \, f(q,p)\, e^{-\ii yp}\,.
\end{equation}
Of course we suppose that the kernel  \eqref{AfintK} and involved partial distributions are well defined, at least in the sense  
of distributions.   
\subsubsection*{Particular functions} 
If $f$ factorises as $f(q,p)= u(q)v(p)$, then the kernel \eqref{AfintK} factorises as 
\begin{equation}
\label{AuvintK}
 \mathcal{A}_{uv}(x,x^{\prime})= \widehat{v}(x^{\prime}-x)\, F
 _u(x,x^{\prime})\,,
 \end{equation}
with 
\begin{equation}
\label{defFxx}
F_u(x,x^{\prime})= \frac{1}{2\pi} \int_{-\infty}^{+\infty}\mathrm{d} q\, u(q)\, \widehat{\Pi}_{p}\left(x-x^{\prime},q- \frac{x+x^{\prime}}{2}\right)\,.
\end{equation}
Hence, in the case $v(p)= p^n$, for a nonnegative integer $n$, standard distribution calculus leads to the formula:
\begin{equation}
\label{Aupn}
 A_{up^n}= \sqrt{2\pi} \sum_{m=0}^n \binom{n}{m} \,(-\ii )^{n-m} \left.\frac{\partial^{n-m}}{\partial {y}^{n-m}}F_u(Q,y)\right|_{y=Q}\,P^{m}\,,
\end{equation}
where we suppose that all manipulations in the above are valid (integrations, derivations,...) at least on the level of distributional calculus,  which entails a supplementary conditions on $\mathfrak{Q}_0$ through \eqref{WHtr}.  

There results that if $f$ depends on $q$ only, $f(q,p)\equiv u(q)$,  its quantization is the multiplication  operator in $\mathcal{H}$ :
\begin{equation}
\label{uq}
A_{u} = \mathsf{w}_u(Q)\, , \quad \mathsf{w}_u(x)=\sqrt{2\pi}  F_u(x,x) = \frac{1}{\sqrt{2\pi }}\,u\ast \overline{\mathcal{F}}[\Pi(0,\cdot)](x)\equiv u\ast\gamma(x)\,,
\end{equation}
where $ \overline{\mathcal{F}}$ is the inverse 1-d Fourier transform, and ``$\ast$'' stands for convolution with respect to the second variable $(\cdot)$. We have introduced in \eqref{uq} the convenient notation:
\begin{equation}
\label{defgam}
\gamma(x):=  \frac{1}{\sqrt{2\pi }}\,\widehat{\Pi}_{p}(0,-x)=  \frac{1}{\sqrt{2\pi }}\,\overline{\mathcal{F}}[\Pi(0,\cdot)] (x)\,,
\end{equation}
where the 1-d Fourier transform concerns the second variable $(\cdot)$. 

If $f(q,p)\equiv v(p)$ is a function of $p$ only, then $A_v$ depends on $P$ only through the convolution:
\begin{equation}
\label{hp}
A_{v} = \mathfrak{v}(P)\, , \quad  \mathfrak{v}(p) = \frac{1}{\sqrt{2\pi }}\,v\ast \mathcal{F}[\Pi(\cdot,0)](p)\equiv v\ast\varpi(p)\, ,
\end{equation}
with
\begin{equation}
\label{varpip}
\varpi(p):= \frac{1}{\sqrt{2\pi }}\,\mathcal{F}[\Pi(\cdot,0)](p)\,. 
\end{equation}

For the simplest cases $u(q)=q$ and $v(p)=p$ one obtains
\begin{equation}
\label{qandp}
A_q = Q  - \ii\left.\frac{\partial}{\partial p} \Pi(0,p)\right|_{p=0}\, , \quad A_p= P + \ii \left.\frac{\partial}{\partial q} \Pi(q,0)\right|_{q=0}\,,
\end{equation}
and so the expected canonical commutation rule $\left[A_q,A_p\right]= \ii \bI$.
This result is actually the direct consequence of the underlying Weyl-Heisenberg covariance when one expresses Eq.\eqref{covtrans1} on the level of infinitesimal generators. Also the additive constants appearing in \eqref{qandp} vanish if moreover  $\Pi$ is even, $\Pi(-q,-p) = \Pi(q,p)$. 

The two following cases have also to be considered in  regard to the content of this paper. 
\begin{align}
\label{Aup1}
 A_{up}&= \mathsf{w}_u(Q) P  + c_u(Q)\, , \quad c_u(x):= -\ii  \sqrt{2\pi} \left.\frac{\partial}{\partial {y}}F_u(x,y)\right|_{y=x}\,,\\
 \label{Aup2}
 A_{up^2}&= \mathsf{w}_u(Q) P^2 +2 c_u(Q)P + d_u(Q)\,, \quad d_u(x):= - \sqrt{2\pi }\,\left.\frac{\partial^2}{\partial^2 y} F_{u}(x,y)\right|_{y=x}\, . 
 \end{align}

\subsubsection*{Particular $\Pi$ or  $\mathfrak{Q}_0$}
\begin{enumerate}
  \item[(i)] If $\Pi= 1$, then $\mathfrak{Q}_0= 2 \sfP$, and $\mFs[\Pi](q,p)= 2\pi \delta(q,p)$.
  \medskip
  \item[(ii)]  If $\mathfrak{Q}_0= |\psi\rg\lg\psi|$, with $\Vert \psi\Vert=1$, then 
 \begin{equation}
\label{psiPiqp}
 \Pi(q,p)= e^{-\ii \frac{qp}{2}}\left(\mathcal{F}[\psi]\ast\mathcal{F}[\mathrm{t}_{-q}\psi]\right)(p)\, , 
\end{equation}
   where $\mathcal{F}[\mathrm{t}_{-q}\psi](p)= e^{\ii qp}\mathcal{F}[\psi](p)$. The corresponding integral kernel is given by
\begin{equation}
\label{Afpsi}
 \mathcal{A}_f(x,x^{\prime})= \frac{1}{\sqrt{2\pi }}\int_{-\infty}^{+\infty}\mathrm{d} q\, \widehat{f}_{\omega}(q,x^{\prime}-x)\, \psi(x-q)\,\overline{\psi(x^{\prime}-q)}
\end{equation}
   The  symplectic Fourier transform of \eqref{psiPiqp} reads
   \begin{equation}
\label{FsPiPsi}
\begin{split}
\mFs[\Pi](q,p)&= \int_{-\infty}^{\infty}\ud y \, e^{-\ii yp}\,\overline{\psi\left(\frac{y}{2}-q\right)}\, \psi\left(-\frac{y}{2}-q\right)\\ &\equiv 2\pi\mathcal{W}_{\psi}(-q,-p)\,,
\end{split}
\end{equation} 
where $\mathcal{W}_{\psi}$ is the Wigner function \cite{wigner32,degosson17}  of the pure state $|\psi\rg\lg\psi|$. It is well known \cite{hudson74,sotoclaverie83} that the latter is a probability distribution on the plane with measure $\ud q\,\ud p$ only if $\psi(x)$ has the form
\begin{equation}
\label{psigauss}
\psi(x)= N\,e^{-ax^2 + bx +c}\,, \quad a,b,c \in \C\,, \quad \mathrm{Re}(a) > 0,
\end{equation}
where $N$ is a suitable normalization constant.

\end{enumerate}

\subsection{Semi-classical portrait}
The quantization formula \eqref{quantPi1} allows to prove  an interesting trace formula (when applicable to $f$). From $\mathrm{Tr}\left(U(q,p)\right)=  2\pi \delta(q,p)$ we obtain 
\begin{equation}
\label{traceAf}
  \mathrm{Tr}\left(A_f\right)=  \overline{\mFs}[f](\vz)=  \int_{\R^2}    f(q,p)\,  \,\frac{\ud q\, \ud p}{2\pi}\, .
\end{equation}
By using \eqref{traceAf} we derive the 
\textit{quantum phase space, i.e.,  semi-classical,  portrait} of the operator as an autocorrelation averaging of the original $f$. 
 More precisely, starting from a function (or distribution) $f(q,p)$ and through its quantum version $A_f$, one defines  the new function $\widecheck{f}(q,p)$ as
\begin{equation}
\label{fmapcf}
\begin{split}
\mathrm{Tr}\left(\mathfrak{Q}(q,p)A_f\right)&=\int_{\R^2}  \, \mathrm{Tr}\left(\mathfrak{Q}(q,p)\,\mathfrak{Q}(q^{\prime},p^{\prime})\right)\, f(q^{\prime},p^{\prime})\frac{\ud q^{\prime}\, \ud p^{\prime}}{2\pi }\\
&:= \widecheck{f}(q,p) \, . 
\end{split} 
\end{equation}
The map $(q^{\prime},p^{\prime})\mapsto  \mathrm{Tr}\left(\mathfrak{Q}(q,p)\,\mathfrak{Q}(q^{\prime},p^{\prime})\right)$ might be a probability distribution if this expression is non negative.  Now, this map is better understood from the equivalent formula,
\begin{equation}
\label{fmapcf}
\widecheck{f}(q,p)  =\int_{\R^2}  \left(\mFs\left[\Pi\right]\ast\mFs\left[\widetilde\Pi\right]\right)(q^{\prime}-q,p^{\prime}-p)\, f(q^{\prime},p^{\prime}) \,\frac{\ud q^{\prime}\,\ud p^{\prime}}{4\pi^2 }\, , 
\end{equation}
where $\widetilde\Pi (q,p):=\Pi (-q,-p)$.
In particular we have for the coordinate functions $\widecheck{q}= q  + q_0$ and $ \widecheck{p}= p + p_0$, for some constants $q_0$ and $p_0$.

Eq. \eqref{fmapcf} represents the convolution ($\sim$ local averaging)  of the original $f$ with the autocorrelation of the symplectic Fourier transform of the (normalised) weight $\Pi(q,p)$. 
Hence,  we are incline to choose windows $\Pi(q,p)$, or equivalently $\mathfrak{Q}_0$,  such that
\begin{equation}
\label{probdist}
 \mFs\left[\Pi\right]
\end{equation}
is a probability distribution on the phase space $\R^2$ equipped with the measure $\dfrac{\ud q\, \ud p}{2\pi}$. 
That  $\mathfrak{Q}_0$ be a density operator, i.e.,
\begin{equation}
\label{Q0dens}
\mathfrak{Q}_0 = \sum_i p_i |\psi_i\rg\lg\psi_i|\, , \quad \Vert \psi\Vert=1\, , \quad 0\leq p_i\leq 1\, , \quad  \sum_i p_i=1\,,
\end{equation} 
is not a sufficient condition \footnote{Contrarily to what it was claimed in \cite{gazeau18}} as it is shown with the pure case state \eqref{FsPiPsi}. 
The condition  is not necessary either, since 
the uniform Weyl-Wigner choice  $\Pi(q,p)= 1$ yields $ \mFs\left[1\right](q,p)= 2\pi \delta(q,p)$ 
and $\mathfrak{Q}_0 = 2\mathrm{P}$, which is not a density operator. Note that $\widecheck f= f$ in this case. 

With a true probabilistic content,  the meaning of the convolution
\begin{equation}
\label{truedist}
 \mFs\left[\Pi\right]\ast\mFs\left[\widetilde\Pi\right]
\end{equation}
is clear: it is the probability distribution  for the difference of two vectors in the phase space plane, viewed as independent random variables,  and thus is perfectly adapted to the abelian and homogeneous structure of the classical phase space: no origin should be privileged.

\section{Quantization and semi-classical portraits of weighted or truncated observables}
\label{qutrunc}

We consider  classical motions which are geometrically
 restricted to hold within some bounded Borel subset $E$,  e.g. the interval $E=[a,b]$,  of the configuration space $\R$.  A standard method to construct a positive smooth approximation of the characteristic (or indicator) function $\chi\se$  of  $E$ is the following (see for instance \cite{donoghue69}). Let us choose a positive function $\omega\se\in C^{\infty}$ which is zero for $x\notin E$ and is normalised in the sense that $\int_{\R}\omega\se(x)\,\ud x= \int\se\omega\se(x)\,\ud x= 1$. For $\sigma>0$ we define 
\begin{equation}
\label{omsig}
 \omega\ses= \frac{1}{\sigma}\omega\se\left(\frac{x}{\sigma}\right)\, .
\end{equation}
The functions $\omega\ses$ are positive, normalised, $C^{\infty}$, and they vanish for $x\notin \sigma E$. Let
\begin{equation}
\label{uEsig}
u\ses(x):= \omega\ses\ast \chi\se(x)= \int\se\omega\ses(x-y)\,\ud y\,. 
\end{equation}
It can be shown that $u\ses\to \chi\se$ as $\sigma \to 0$ almost everywhere, and, moreover, that the convergence is uniform on $\interior{E}$. 
 
We then smoothly  truncate all classical observables to $E$ 
\begin{equation}
\label{trunc}
f(q,p)\mapsto u\ses(q)f(q,p)\equiv f\ses(q,p)\,  .
\end{equation}
In particular, the original canonical coordinates $q$ and $p$ become the truncated observables $u\ses(q)q$ and $u\ses(q)p$ respectively.

We further proceed with the Weyl-Heisenberg  quantization \eqref{qmap} or \eqref{quantPi1}  of the truncated   observables $f\ses$, and obtain the $(E,\sigma)$-modified operator,
\begin{equation} 
\label{quantchif1}
f\ses(q,p)\mapsto A_{f\ses} 
=  \frac{1}{2\pi }\int_{\R^2}  \ud q\,\ud p\,  \Pi(q,p)\,  \overline{\mFs}[f\ses](q,p)\,U(q,p)\,.
\end{equation}
Quantization formulae given in Section \ref{CIQWH} apply here with the change $f\mapsto  f\ses$.
Let us start with the  quantization of the case $f=1$, i.e.,  $f\ses(q,p)= u\ses(q)$. It yields  the ``window'' or ``localisation'' \cite{hewong96} multiplication operator: 
\begin{equation}
\label{restId1}
 A_{u\ses} = (u\ses\ast\gamma)(Q)= \int_\R \ud y\, \gamma(Q-y)\,u\ses \equiv\mathsf{w}_{u\ses}(Q)\, ,
\end{equation}
with the notation \eqref{defgam}. The window function $ \mathsf{w}\se(x)$ will play a fundamental r\^ole in the sequel. Its Fourier transform is given by
  \begin{equation}
\label{hatwE}
\hat{\mathsf{w}}\ses(p)= \sqrt{2\pi}\,\hat{u}\ses(p)\,\hat{\gamma}(p)= \hat{u}\ses(p)\,\Pi(0,-p)\,. 
\end{equation}
For our present purpose, we also give the quantum counterparts of the $(E,\sigma)$-truncated position, momentum and momentum squared ($\sim$ kinetic energy). They are deduced from \eqref{Aup1} and  \eqref{Aup2}, and read as:
\begin{equation}
\label{Eqchi}
 A_{q\ses}= \mathsf{w}\ses(Q)Q + b\ses(Q)\, , \quad b\ses(x) = (u\ses\ast Q\gamma)(x)\,, 
 \end{equation}
 \begin{equation}
 \label{Epchi} A_{p\ses} = \mathsf{w}\ses(Q) P + c\ses(Q)\,, \quad c\ses(x) = -\ii \sqrt{2\pi}\,\left.\frac{\partial}{\partial y} F_{u\ses}(x,y)\right|_{y=x}\,,
 \end{equation}
  \begin{equation}
  \label{Ep2chi} A_{p^2\ses}   = \mathsf{w}\ses(Q) P^2 +2 c\ses(Q)P + d\ses(Q)\,, \quad d\ses(x)= - \sqrt{2\pi }\,\left.\frac{\partial^2}{\partial^2 y} F_{u\ses}(x,y)\right|_{y=x}\,.
\end{equation}
 The modification of the basic commutation rule follows:
\begin{equation}
\label{ccpqE}
\left[A_{q\ses} ,A_{p\ses}\right] = \ii\left[ \mathsf{w}^{\prime}\ses(Q) Q + \mathsf{w}\ses(Q) + \mathsf{w}\ses(Q) b^{\prime}\ses(Q)\right]\,,
\end{equation}
which corresponds to a deformed version of the uncertainty inequality for these operators.

Of course, the validity of the above formulas depends on the regularity   properties of the function $\widehat{\Pi}_{p}(x,y)$ which appears in the expression \eqref{AfintK} of the integral kernel of $A_f$. From now on we make the following minimal assumptions.
\begin{assum}
The apodisation function (or tempered distribution) $\Pi(q,p)$ is chosen such that its symplectic Fourier transform,
\begin{equation}
\label{probdist2}
 \mFs\left[\Pi\right](q,p)=\int_{\R^2}\frac{\ud q^{\prime}\,\ud p^{\prime}}{2\pi}\,e^{-\ii(qp^{\prime}-pq^{\prime})}\,\Pi(q^{\prime},p^{\prime})\,,
\end{equation}
is  non-negative and so, from the normalisation $\Pi(0,0)=1$,   is a probability distribution on $\R^2$ equipped with the measure $\frac{\ud q\,\ud p}{2\pi}$.  
\end{assum}
\begin{assum} The partial Fourier transform of $\Pi(q,p)$, 
\begin{equation}
\label{partFPi}
\widehat{\Pi}_{p}(q,y)= \frac{1}{\sqrt{2\pi}}\int_{-\infty}^{+\infty}\ud p\,e^{-\ii yp}\,\Pi(q,p)\, , 
\end{equation}
 is at least twice  continuously differentiable on $\R^2$. 
\end{assum}
\begin{assum} The partial Fourier transform $\widehat{\Pi}_{p}(q,y)$ is non-negative at $q=0$:
\begin{equation}
\label{nnegpart}
\widehat{\Pi}_{p}(0,y)\geq 0\ \forall y\in \R\,. 
\end{equation}
 \end{assum}
 Note that it results from Assumptions \eqref{partFPi} and \eqref{nnegpart} that the window function $\mathsf{w}\ses(x)$, given by the convolution \eqref{restId1}, is non-negative and goes to $0$ as $x\to \pm\infty$. 
 
Also note that the  Hilbert space in which act the above ``$(E,\sigma)$-modified'' operators  is left unchanged. Thus, in position representation, one continues to deal with  $\mathcal{H}= L^2(\R)$.
Nevertheless, our approach gives rise to a smoothing of the constraint boundary $\partial E$, 
i.e., a ``fuzzy'' boundary, and also, if the function $\gamma$ or $\Pi(0,\cdot)$ is smooth enough,   a smoothing of all  restricted observable $f\ses(q,p)$ introduced in the quantization map \eqref{quantchif1}, including the limit at $\sigma=0$, i.e, when $u\sez =\chi\se$. 
Indeed, there is no mechanics outside the set $\sigma E\times\R$ defined by the position constraint  on the classical level. 
It is however not the same on the quantum level since our quantization method allows to go beyond the boundary of this set in a  way which can be smoothly rapidly decreasing, depending on the function $\Pi(0,p)$.

Consistently,  the semi-classical phase space portrait of the operator \eqref{quantchif1} is given by \eqref{fmapcf}. An equivalent form of the latter reads as more condensed:
\begin{equation}
\label{semclassA}
\widecheck{f}\ses(q,p)  = \frac{1}{2\pi } \left(\overline{\mFs}\left[\Pi\,\widetilde\Pi\right]\ast f\ses\right)(q,p) \, . 
\end{equation}
It should be found to be concentrated on the classical $E\times \R$, and so viewed as a new classical observable defined on the full phase space $\R^2$ where  $q$ and $p$ keep their status of canonical variables. 

Thus, we have the  sequence\\
 \begin{equation}
\label{regseq}
 \begin{split}
\mathrm{virtual} \ f(q,p) \rightarrow \  \mathrm{truncated} \ &f\ses(q,p)\\
&\symbolwithin{\downarrow}{=} \\
 \mathrm{regularised} \ \widecheck{f}\ses(q,p)  \leftarrow\  \mathrm{quantum}\  &A_{f\ses}\,,
\end{split}
\end{equation}\\
allowing to establish a semi-classical dynamics \textit{\`a la} Klauder \cite{klauder12,klauder15},  mainly concentrated on $ E\times\R$ as $\sigma \to 0$. In other words, $\widecheck{f}\ses(q,p)$ is a different regularization of the original $f(q,p)$, which takes the place of $f\ses(q,p)$.

Note that we can confine our study to the elementary case of the interval,
\begin{equation}
\label{Eint}
E= (\alpha,\beta)\,,
\end{equation}
since these intervals generate the $\sigma$-algebra of Borel sets of the real line. 
This particular case allows us  to compare the well-known quantum mechanics of a free particle moving in the infinite square well with boundaries the point $\alpha$ and $\beta$ with the new functional material described in the above. Then  we have to compare the Hilbert space $L^2(\alpha,\beta)$ with the range $\mathcal{R}\se$ of the bounded positive multiplication operator $\sqrt{\mathsf{w}\ses(Q)}$ defined by the square root of the window operator:
\begin{equation}
\label{Winop}
L^2(\R) \ni \phi(x) \mapsto \sqrt{\mathsf{w}\ses(x)}\phi(x) \in \mathcal{R}\se \subset L^2(\R)\,. 
\end{equation}
The closure $\overline{\mathcal{R}\se}$  of this range is a sub-Hilbert space of $L^2(\R)$, which is itself a subspace of $L^2(\R,\mathsf{w}\se(x)\ud x)$.

\section{The question of self-adjointness}
\label{windof}
\subsection{Prelude}
Through the material presented above we are faced to the following  situation.  
Our quantization procedure has yielded  the non-negative window   multiplication operator $\mathsf{w}_{u\ses}(Q)$ given by: 
 \begin{equation*}
\label{restId1}
\mathsf{w}_{u\ses}(Q) =  A_{u\ses} = (u\ses\ast\gamma)(Q)= \int_\R \ud y\, \gamma(Q-y)\,u\ses(y) \, ,
\end{equation*}
with $\gamma(x)= \frac{1}{\sqrt{2\pi }} \overline{\mathcal{F}}[\Pi(0,\cdot)](x)$.
The quantisation of the $(E,\sigma)$-truncated position yields the symmetric bounded, i.e. self-adjoint multiplication operator:
\begin{equation*}
 A_{q\ses}= \mathsf{w}\ses(Q)Q + b\ses(Q)\, , \quad b\ses(x) = (u\ses\ast Q\gamma)(x)\,.
 \end{equation*}
The quantisation of the $(E,\sigma)$-truncated momentum yields the symmetric operator:
\begin{equation*}
  A_{p\ses} = \frac{1}{2}\left\{\mathsf{w}\ses(Q) , P\right\}  + \frac{\ii}{2} \mathsf{w}^{\prime}\ses(Q)+  c\ses(Q)\,, \quad c\ses(x) = -\ii \sqrt{2\pi}\,\left.\frac{\partial}{\partial y} F_{u\ses}(x,y)\right|_{y=x}\,,
\end{equation*}
where 
$F_u(x,x^{\prime})= \frac{1}{2\pi} \int_{-\infty}^{+\infty}\mathrm{d} q\, u(q)\, \widehat{\Pi}_{p}\left(x-x^{\prime},q- \frac{x+x^{\prime}}{2}\right)$.
The quantisation of the $(E,\sigma)$-truncated momentum squared ($\sim$ kinetic energy) yields the symmetric operator:
\begin{equation*}
  A_{p^2\ses}   = \frac{1}{2}\left\{\mathsf{w}\ses(Q), P^2\right\} +  \frac{\ii}{2} \mathsf{w}^{\prime\prime}\ses(Q) +\left\{c\ses(Q),P\right\} + \frac{\ii}{2} c^{\prime}\ses(Q) + d\ses(Q) \,,
\end{equation*}
with $d\ses(x)= - \sqrt{2\pi }\,\left.\frac{\partial^2}{\partial^2 y} F_{u\ses}(x,y)\right|_{y=x}$.
 
 {Note that the  purely $Q$-dependent terms appearing in the expressions of $A_{p\ses} $ and $A_{p^2\ses}$ 
are bounded multiplication operators and do not play any significant role in the self-adjointness properties of these two operators.}  

Hence, we are led to examine the fundamental questions: Are self-adjointness of $ P$ and  $P^2$ preserved under such regularisations, or, equivalently, are the symmetric operators
 \begin{equation}\label{41}
 P \mapsto  \sqrt{\mathsf{w}\ses(Q)} P \sqrt{\mathsf{w}\ses(Q)}\, , \quad  P^2 \mapsto  \sqrt{\mathsf{w}\ses(Q)} P^2 \sqrt{\mathsf{w}\ses(Q)}
\end{equation}
self-adjoint? This is the problem we will consider, at a rather general level, in the rest of the paper.
\subsection{Mathematical setting and results}
Let us examine this question on the level of standard operator analysis.  Let $\calH$ be a Hilbert space. If $X$ is a linear operator, we denote by $D(X)$, $N(X)$, $R(X)$, the domain, the kernel and the image of $X$, respectively.

Assume that $A$ is a bounded everywhere defined symmetric operator and $P$ a self-adjoint operator with domain $D(P)$, not necessarily coinciding with the momentum operator.

We are interested in determining conditions for $APA$ to be self-adjoint. Of course, in view of (\ref{41}), $A$ will be identified later with $\sqrt{\mathsf{w}\ses(Q)}$, while $P$ will be identified with the momentum operator, or with its square.

\medskip
Clearly, $D(APA)=D(PA)= \{x\in \calH: Ax\in D(P)\}$.

It is easily seen that the density of $D(PA)$ and the boundedness of $A$ imply that $(AP)^*=PA$ and therefore $AP$ is closable and $\overline{AP}=(PA)^*$.

\begin{lemma}\label{prop_one} Let $PA$ be densely defined. The following statements hold.
	\begin{itemize}
	\item[(i)] $PA$ is closed and $(APA)^*= (PA)^*A$.
	\item[(ii)]If $AP$ is closed, then $APA$ is self-adjoint.
	\item[(iii)] If $A$ has an everywhere defined bounded inverse,  then $APA$ is self-adjoint. 
\end{itemize}	
	\end{lemma}
\begin{proof} (i): The operator $PA$ is closed. Indeed, if $x_n\to x$ and $PAx_n \to y$, then since $Ax_n\to Ax$ it follows that $Ax \in D(P)$ and $PAx_n \to PAx$. 
The density of $D(PA)$ also implies that
	 $(APA)^*$ exists and we have
	\begin{equation}\label{thisone} y\in D((APA)^*) \Leftrightarrow\ip{(APA)x}{y}= \ip{PAx}{Ay}= \ip{x}{(APA)^*y}\end{equation}
	The second equality holds if and only if  $Ay\in D((PA)^*)$ and $(APA)^*y= (PA)^*Ay$.
	
	(ii): If $AP$ is closed then,  $AP= \overline{AP}= (PA)^*$. This implies that $APA=(PA)^*A= (APA)^*$, by (i).
	
	(iii): If $A^{-1}$ exists in ${\mc B}(\calH)$,  $AP$ is closed, as it is easily checked.
\end{proof}

Thus, the self-adjointness of $APA$ is guaranteed if one of the equivalent conditions (a): $AP$ closed or (b): $(PA)^*=AP$, is satisfied.

%
%
\berem Let us suppose that $A$ is injective but has a densely defined unbounded inverse $A^{-1}$ i.e., $0\not\in \rho(A)$, the resolvent of $A$. The operator $A^{-1}$ is self-adjoint. The study of the closedness of $AP$ is more complicated in this case, since $A$ is not bounded from below.
	Assume that $x_n \to x$ and $APx_n \to y$, then we can't conclude that $\{Px_n\}$ is convergent without further assumptions. But if $\{Px_n\}$ converges to some $z\in \calH$, we have
	$$ APx_n \to y \; \mbox{and} \; Px_n \to z.$$
	Since $A^{-1}$ is closed we obtain that $y\in D(A^{-1})$ and $z= A^{-1}y$.
	The closedness of $P$ implies that $x\in D(P)$ and $z=Px$. Then, $A^{-1}y=z=Px$ and, finally $y=APx$.

	\enrem

\beex
Let us suppose that $D(PA)=D(P)$ and $PAf-APf=Bf$, for every $f\in D(P)$, with $B$ bounded and $BD(P)\subset D(P)$.
  We prove that $AP$ is closed. Indeed, let $f_n\to f$ with $APf_n \to g$ Then $Af_n \to Af $ and $PA f_n-APf_n =Bf_n$ converges because of the boundedness of $B$. Hence, $\{PAf_n\}$ is convergent. Since $PA$ is closed , we get $f\in D(PA)=D(P)$ and $PAf= \lim_{n\to \infty}PAf_n$. Therefore, $APf_n \to PAf +Bf$. Thus $f\in D(AP)$ and $g= PAf+Bg$.
  Then $AP$ is closed and, by (ii) of Lemma \ref{prop_one}, $APA$ is self-adjoint.
\enex
We now turn to the main problem.

\begin{thm}\label{thm7} Suppose that $R(A)$ is closed and ${N(A)}\subset D(P)$, $PN(A)\subset N(A)$. Then $AP=(PA)^*$.
\end{thm}
\begin{proof} First notice that the assumptions imply that $\calH=N(A)\oplus R(A)$. Moreover, since $R(A)$ is closed, there exists $\gamma>0$ such that
$$ \|Ay\| \geq \gamma \|Qy\|$$
where $Q$ denotes the projection operator onto $R(A)$ (consequence of \cite[Theorem 4.13]{rudin}).
Let $z\in D((PA)^*)$ and $y \in D(P)\cap R(A)$, then $y=Ax$ for some $x\in \calH$. The element $x$ can be written as $x= (x-x_0)+x_0$ with $x_0\in N(A)$ and $x-x_0\in R(A)$ and clearly $y=A(x-x_0)$. Hence
$$ \ip{Py}{z}=\ip{PA(x-x_0)}{z}=\ip{x-x_0}{(PA)^*z}$$
and so
\begin{align*} |\ip{Py}{z}|&\leq \|x-x_0\|\|(PA)^*z\|=\|Q(x-x_0)\|\|(PA)^*z\|\\&\leq \gamma^{-1}\|Ax\|\|(PA)^*z\|=  \gamma^{-1}\|y\|\|(PA)^*z\|.\end{align*}
Every element of $D(P)$ is the sum of an element of $N(A) \subset D(P)$ and one of $R(A)\cap D(P)$. The operator $P$ is bounded on $N(A)$.
Thus,
$$|\ip{Py}{z}| \leq C \|y\|, \forall y\in D(P).$$ Therefore $z\in D(P)$ and in conclusion $(PA)^*=AP$.
\end{proof}

It is easy to prove that
\begin{prop}\label{Prop6}
$R(A)$ is closed if and only if  there exists $m>0$ such that $$\|Ax\|\geq m \|x\|, \quad \forall x \in \calH.$$
\end{prop}


\bigskip
{
\begin{prop} \label{nicecore} If  $AD(PA)$  contains a core $\D$ for $P$, then $APA$ is self-adjoint.
\end{prop}
\begin{proof}
Let $y \in D((APA)^*)$. Then there exists $y^*\in \calH$ such that
$$\ip{APAx}{y}= \ip{x}{y^*}, \quad \forall x \in D(PA).$$
Now,
$$\ip{z}{y^*}= \ip{A^{-1}Az}{y^*}, \quad \forall z \in \calH,$$
and so $y^*\in D(A^{-1})$.
Then,
$$\ip{APAx}{y}=\ip{PAx}{Ay}= \ip{x}{y^*}=\ip{A^{-1}Ax}{y^*}=\ip{Ax}{A^{-1}y^*}, \quad \forall x \in D(PA).$$
The  equality
$$\ip{PAx}{Ay}= \ip{Ax}{A^{-1}y^*}, \quad \forall x \in D(PA)$$
implies that $Ay \in D((P\upharpoonright_{AD(PA)})^*)$. \\
If $AD(PA)$  contains a core $\D$ for $P$, \mbox{$D((P\upharpoonright_{AD(PA)})^*)=D(P)$},
and therefore $APA$ is self-adjoint.
\end{proof}
 }

\section{Modifying the momentum operator}
\label{modmom}
With the above results at hand,  we now examine the cases encountered in Eq. \eqref{Epchi} and in (\ref{41}). 

  Let $\calH= L^2({\mb R})$. Let $P= -\ii\dfrac{\ud}{\ud x}$, defined on $D(P)=\sob$, where, for $\mathcal{D}\subset \R$,  $W^{k,p}(\mathcal{D}): = \left\{f\in C^k(\mathcal{D})\,  \mid \, \Vert f\Vert_{k,p}= \left(\sum_{j=0}^k\Vert f^{(j)}\Vert_p^p\right)^{1/p}< \infty\right\}$. It is well-known that $P$ is self-adjoint on $D(P)$, and  that $C_c^\infty(\mb R)$ (or the Schwartz space ${\mc S}(\mb R)$) is a core for $P$. Let $A$ be the multiplication operator by a measurable  essentially bounded real valued function $a(x)$. 
Thus,
$(Af)(x)=a(x)f(x)$ for every $f\in \ltwo$. Clearly $A$ is bounded and $\|A\|=\|a\|_\infty$.

By Proposition \ref{Prop6},  $R(A)$ is closed if, and only if, $\inf_{x\in {\mb R}} a(x)>0$. Therefore, if $a(x)$ is a positive function with $\lim_{|x|\to \infty}a(x)=0$, $R(A)$ is not closed and Theorem \ref{thm7} cannot be applied.

We remind that the following statements are equivalent \cite[Sec. 4.1, Example 1]{weidmann}
\begin{enumerate}
\item $R(A)$ is dense.
\item $a(x)\neq 0$ almost everywhere in ${\mb R}$
\item $A$ is injective.
\end{enumerate}
In this case, the inverse $A^{-1}$ is the multiplication operator defined by the function
$$a_1(x)= \left\{\begin{array}{cl} a(x)^{-1} & \mbox{ if } a(x)\neq 0 \\ 0  &\mbox{ if } a(x)= 0. \end{array}\right.$$

\medskip
We suppose now that $a$ is bounded and sufficiently smooth in the  sense that $a$ is a member of the following class of functions

$${\mc A}:= \{a \in L^\infty (\mb R): a\phi \in \sob, \, \forall \phi \in C^\infty_c(\mb R) \}.$$
 If $a\in {\mc A}$ then the operator $PA$  is densely defined, since $D(PA)$ contains $C^\infty_c(\mb R)$. 
 
It is easily seen that $C^1(\mb R) \cap L^\infty(\mb R) \subset {\mc A}$ and
$\sob \cap L^\infty(\mb R) \subset {\mc A}$
 
On the other hand, the operator $AP$ is densely defined for every $a\in L^\infty(\mb R)$; indeed, $D(AP)=D(P)= \sob$.

\berem \label{rem_classA} Using the density of $C^\infty_c(\mb R)$ in the Hilbert space $\sob$ one can prove that if $a\in {\mc A}$ then $af\in \sob$ for every $f\in \sob$.\enrem 
We examine here the problem of the closedness of the operator $AP$    when $a$ is nonnegative function $a(x)$ and $a\in C^1(\mb R) \cap L^\infty(\mb R)$. We follow in this concrete case the strategy outlined in Section \ref{windof}.

\medskip
Let $\{f_n\}$ be a sequence in $D(AP)=\sob$ such that $f_n \to f$ and $APf_n \to g$, with respect to the $L^2$-norm. We want to investigate under which conditions on $a(x)$ we can conclude that $f\in D(AP)$ and $g=APf$. 

Let $\epsilon >0$ and let $$F_0=\{x\in {\mb R}: a(x)>\epsilon \}.$$
Then,
$$\int_{F_0} \left| -\ii a(x) f_n'(x) + \ii a(x) f'_m(x) \right|^2\ud x \geq \epsilon^2 \int_{F_0}  \left|  f_n'(x) - f'_m (x) \right|^2\ud x. $$
Therefore, the sequence $\{ f_n'\}$ converges in ${ L}^2(F_0)$. Let $h$ be its limit. Since \mbox{$\|APf_n- g\|\to 0$}, there exists a subsequence $\{f_{n_k}\}$ such that $-\ii a(x)f_{n_k}'(x)$ converges to $g(x)$ a.e. Hence we conclude that $h(x)=\dfrac{g(x)}{-\ii a(x)}$ almost everywhere in $F_0$.\\
So if there exists $\epsilon >0$ such that ${F_0}={\mb R}$, we deduce, taking into account that $P$ is closed,  that $f\in \sob$ and $APf =g$; that is, $AP$ is closed. This is not surprising because the equality ${F_0}={\mb R}$ implies that the multiplication operator $A$ is bounded with bounded inverse. The same conclusion holds if ${\mb R}\setminus {F_0}$ is a null-set rather than the empty set.

 \subsection{Self-adjointness of $APA$} 
Given a bounded measurable function $a$ on $\mb R$, we define the modified momentum operator $P_a$ as the operator acting as
$$ (P_af)(x)=-\ii a(x) \frac{\ud}{\ud x}(a(x)f(x)).$$
As seen before if $a\in C^1(\mb R) \cap L^\infty(\mb R)$ then $P_a$ is densely defined and symmetric.
We will study the self-adjointness of $P_a$.
As a first step, we compute $(PA)^*$, because if we can prove that $(PA)^*=AP$ then from (i) of Lemma \ref{prop_one} it follows that $P_a$ is self-adjoint.

 Observe that $g\in D((PA)^*)$ if and only if there exists $h\in \ltwo$ such that, for every $f\in D(PA)$,
\begin{align*}\ip{PAf}{g} &=\int_{\mb R} -\ii \frac{\ud}{\ud x} (a(x)f(x))\cdot\overline{g(x)} \ud x \\
&= \int_{\mb R} f(x) \overline{h(x)}\ud x.
\end{align*}
If $g \in \sob$, then we can write
\begin{align*}\ip{PAf}{g} &=\int_{\mb R} -\ii \frac{\ud}{\ud x} (a(x)f(x))\cdot\overline{g(x)} \ud x \\
&=  -\ii(a(x)f(x))\cdot\overline{g(x)}\left|_{-\infty}^\infty \right. +\ii \int_{\mb R} a(x)f(x) \overline{g'(x)}\ud x.\\
&=  \int_{\mb R} f(x)\cdot (\overline{-\ii a(x)g'(x)})\ud x,
\end{align*}
where $g'$ denotes the weak derivative of $g$.\\
The first term in the integration by parts is $0$ because the product of functions in $\sob$ is in $\sob$ and if $u\in \sob$ then $\ds \lim_{|x|\to \infty}u(x)=0$ \cite[Corollaries VIII.8, VIII.9]{brezis}.\\
This proves that
$$\{g\in \sob:\, ag' \in \ltwo\} \subseteq D((PA)^*).$$
We notice that, since $a \in L^\infty({\mb R})$ and $g\in \sob$, then, automatically, $ag'\in \ltwo$.
Thus, in conclusion

 \begin{lemma} Let $a\in {\mc A}$. Then
 $ \sob \subset D((PA)^*)$ and $(PA)^*g= -\ii ag'$, for all $g\in\sob$.
 \end{lemma}


\medskip
In general, we do not know if $D((PA)^*)=\sob$.
Thus, we begin with considering the following set of functions:
$$ {\mc F}= \{ a\in {\mc A}:\; D((PA)^*)=\sob\}.$$
The set ${\mc F}$ is nonempty, since it contains every nonzero constant function.
If $a \in {\mc F}$, then
$$D((PA)^*A)=\{ f\in \ltwo: af \in \sob\}=D(APA).$$
But as seen in Section 1, $(APA)^*= (PA)^*A$. In conclusion,
\begin{prop}
	If $a\in {\mc F}$, the operator $APA$ is self-adjoint.
\end{prop}

The next step consists in giving conditions for a function $a\in C^1(\mb  R)\cap  L^\infty({\mb R})$ to belong to ${\mc F}$.

\beex\label{prop9} If $A$ has a bounded inverse and  $C_c^\infty(\mb R) \subset R(PA)=AD(PA)$, then $a\in {\mc F}$.
\enex

In general $A^{-1}$ (if it exists) is unbounded. In this case we get what follows. 
	\begin{prop} \label{prop_ubd} If $A$ has an inverse $A^{-1}$ and  $C_c^\infty(\mb R) \subset R(PA)$, then
		every $g\in D((PA)^*$ admits a regular distributional derivative $g'$.
	\end{prop}
	\begin{proof} The assumption implies that $A^{-1}$ is densely defined. For every $\phi\in C_c^\infty(\mb R)$ there exists a unique $f\in D(PA)$ such that $\phi=Af$. Hence, if $h=(PA)^*g$,
		\begin{align*}\ip{P\phi}{g} &=\ip{PAf}{g} =\int_{\mb R} -\ii \frac{\ud}{\ud x} (a(x)f(x))\cdot\overline{g(x)} \ud x \\
			&=  \int_{\mb R} f(x)\overline{h(x)}\ud x =\ip{A^{-1}\phi}{(PA)^*g}\\
			&= \int_{\mb R}\phi(x)\overline{a(x)^{-1}h(x)}\ud x
		\end{align*}
		This implies that the  distributional derivative $g'$ of $g$  (which exists since $g\in \ltwo$) is a regular distribution and $g'=-a^{-1}h$.
	\end{proof}
	\berem From the previous proof
	we don't get $g'\in \ltwo$; so, we can't conclude that $g\in \sob$. But of course, $ag'\in \ltwo$.
	\enrem
	\begin{prop} \label{prop14} Under the assumption of Proposition \ref{prop_ubd}, we have
		\begin{equation}\label{eq_conc} D((PA)^*)= \{ g\in \ltwo: \, g' \in  \ltwo\}=\sob.\end{equation}
	\end{prop}

\subsection{Further analysis} 

As seen before, when looking for the adjoint of a differential operator one tries to use integration by parts, which holds for functions in 
	$W^{1,2}(\mathbb R)$ (\cite{brezis}, Cor. VIII.9). For finding the adjoint of $P_a:= -\ii a(x)\frac{\ud}{\ud x} (a(x)\cdot)$ we should determine the functions $g\in \ltwo$ for which there exists a function $h\in \ltwo$ such that
	\begin{equation} \label{weakder}
\int_{\mathbb{R}} a(x) (a(x)f(x))' \overline{g(x)}\ud x= \int_{\mathbb{R}}f(x)\overline{h(x)}\ud x, 
\end{equation}
	and for using integration by parts we need to show first that we are dealing with functions in $W^{1,2}(\mathbb R)$. We begin our analysis by considering this question.
	
	\smallskip
	{\bf Question}: does equality \eqref{weakder} imply that  $a(x){ g(x)}$ has a weak derivative in $\ltwo$; i.e. $a(x){ g(x)}\in W^{1,2}(\mathbb R)$?
	
	Taking into account that $a(x)>0$ we rewrite
	$$\int_{\mathbb{R}} (a(x)f(x))' a(x)\overline{g(x)}\ud x= 
	\int_{\mathbb{R}}a(x)f(x)\frac{\overline{h(x)}}{a(x)}\ud x.$$
	If $\{af; f\in D(P_a)\} \supset C_c^\infty (\mathbb R)$, then from the previous equation we obtain
	$$\int_{\mathbb{R}} \phi'(x) a(x)\overline{g(x)}\ud x= 
	\int_{\mathbb{R}}\phi(x)\frac{\overline{h(x)}}{a(x)}\ud x, \quad \forall \phi \in C_c^\infty(\mathbb R)$$
	which implies that $a(x)g(x)$ has a distributional derivative  $(a(x)g(x))'=\dfrac{\overline{h(x)}}{a(x)}$.  If this derivative is in $\ltwo$, then  $ag \in W^{1,2}(\mathbb R)$. In this case, $g\in D(P_a^*)$ and $(P_a^*g)(x)=h(x)= a(x) (a(x)g(x))'$.

	\smallskip Therefore we obtain the following  concrete realization of Proposition \ref{nicecore}. 
	\begin{prop} \label{11}
		If $C_c^\infty (\mathbb R)\subset R(PA)$, then $P_a$ is self-adjoint.
	\end{prop}

\begin{proof}As shown before the integration by parts is allowed. Then we only need a standard calculation, taking into account (see Remark \ref{rem_classA}) that if $a\in {\mc A}$ and $f(x)\in \sob$, then $\lim_{|x|\to \infty}a(x)f(x)=0$. Moreover from the previous discussion it follows that we have to require that $(ag)'$ is an $L^2$-function.
	Then we conclude that if $\{af; f\in D(P_a)\} \supset C_c^\infty (\mathbb R)$,
	then $$D(P_a^*)= \{g \in L^2(\mb R): ag \in \sob \}$$ and this is exactly the domain of $P_a$.
	
	\end{proof}

	\berem We notice that Proposition \ref{11} can be deduced directly from Proposition \ref{prop14}.
	\enrem

	Let $a\in { \mc A}$ and suppose that $a>0$, Then the operator $A^{-1}$ exists and it is possibly unbounded. We prove that in this case $ C_c^\infty (\mathbb R) \subset R(PA)$. Indeed, for every $\phi \in C_c^\infty (\mathbb R)$ the function $h:=\dfrac{\phi}{a}$ belongs to $D(PA)$, since $ah=a\dfrac{\phi}{a}=\phi \in W^{1,2}(\mathbb R)$. Therefore the condition $\{af; f\in D(P_a)\} \supset C_c^\infty (\mathbb R)$ is satisfied. Thus we get

	\begin{thm} \label{prop17} Let $a(x)\in {\mc A}$; $a>0$ almost everywhere.
		 Then the operator $P_a$ is self-adjoint.
	\end{thm}

Under the same conditions we can prove that  operator $P^2_a:= -a(x)\frac{\ud^2}{\ud x^2} a(x)$ is self-adjoint.

Let us now see how the window function used in \cite{gazkoi19}, obtained from $u\sez=\chi\se$, $E=(\alpha,\beta)$, through  coherent state quantization, i.e., when $|\psi\rg$ in \eqref{psiPiqp} is given by the Gaussian $\psi(x)=   \left( \frac{1}{\pi} \right)^{1/4} e^{-x^2/2}$,   fits in our scheme (up to an irrelevant factor).

\beex \label{exone}  Let $a(x)=\int_{x-\beta}^\infty e^{-t^2}\ud t -\int_{x-\alpha}^\infty e^{-t^2}\ud t$, with $\alpha<\beta$.
Clearly $a \in C^\infty(\mb R)$ and one can also write
$$a(x)=\int_{x-\beta}^{x-\alpha} e^{-t^2}\ud t= \int^{\beta}_{\alpha} e^{-(t-x)^2}\ud t\, .$$
Observe that $a(x)>0$ for every $x\in {\mb R}$.

A simple estimation shows that
$$ a(x)=\int_{x-\beta}^{x-\alpha} e^{-t^2}\ud t \leq \beta-\alpha;$$
Hence, $a\in L^\infty (\mb R)$.
\begin{equation}\label{eq_deriv}a'(x)= -e^{-(x-\beta)^2}+ e^{-(x-\alpha)^2}\geq 0\end{equation} when $x\leq  \frac{\beta +\alpha}{2}$; so that $a$ attains its maximum in $ \frac{\beta +\alpha}{2}$. On the other hand, $\ds \inf_{x\in {\mb R}}a(x)=0$. Hence the function $\ds \frac{1}{a(x)}$ is everywhere defined in ${\mb R}$ but it is unbounded. Therefore $A^{-1}$ exists but is unbounded.
The domain of $PA$ is dense since it contains $C_c^\infty(\mb R)$.
The function $a$ is also in $\ltwo$ as well as its derivative; so $a\in \sob$. 
Finally, we observe that $a\in {\mc A}$. Proposition \ref{prop17} applies then ve can conclude that the corresponding operator $P_a$ is self-adjoint.
\enex

{\berem As mentioned at the very beginning of this paper, the momentum operator for a particle moving in a bounded interval $[a,b]$ (with Dirichlet boundary conditions) is not (essentially) self-adjoint. It is clear that we can regard this operator as being of the form $APA$ where $P$ is, as above, the momentum on the real line and $A$ is the multiplication operator by the characteristic function of $[a,b]$. It is clear as well that this case is not covered by any assumption we have made on the function $a$ used for {\em compressing} the operator $P$. 
\enrem}

\section{Conclusion}
\label{conclu}
In this work we have examined the question of the stability of self-adjointness when a self-adjoint operator $P$ is replaced with $APA$, where $A$ is a bounded everywhere defined symmetric operator. This problem is raised within the context of Weyl-Heisenberg covariant integral quantization of truncated classical observables, for instance by replacing the momentum $p$ of a particle moving on the line with its truncated version $\chi\se (q)p$, where $E= (\alpha,\beta)$, and it is connected to the analysis of non-Hermitian Hamiltonians, \cite{benbook,bagbook,bagspringer}. We have proved that in this case the quantum version $APA$ of $\chi\se (q)p$ is self-adjoint. The natural next step to be considered is to determine spectral features  of $APA$ with regard to the spectrum of some self-adjoint extension of the momentum operator  of a particle constrained to move within a bounded set in $\R$, like   the interval $(\alpha,\beta)$, and more generally to compare  the evolution operators $e^{-\ii P^2}$ in $L^2((\alpha, \beta))$ and $e^{-\ii AP^2A}$ in $L^2(\R)$.

\subsection*{Acknowledgments}
J.-P. Gazeau is indebted to the Dipartimento di Matematica e Informatica,
Universit\`a di Palermo, and to the Gruppo Nazionale di Analisi Matematica  la Probabilit\`a e le loro Applicazioni of Indam, for financial support.  F. B. acknowledges the University of Palermo, and the Gruppo Nazionale di Fisica Matematica of Indam. C. T. acknowledges the University of Palermo, and the Gruppo Nazionale di Analisi Matematica  la Probabilit\`a e le loro Applicazioni of Indam.

\subsection*{Conflict of interest} On behalf of all authors, the corresponding author states that there is no conflict of interest.


\begin{thebibliography}{99}
\bibitem{apodi} An apodization function (also called a tapering function or window function) is a function used to smoothly bring a sampled signal down to zero at the edges of the sampled region.
\end{thebibliography}
\end{document}